\documentclass[usenames,dvipsnames]{article}
\usepackage[utf8]{inputenc}
\usepackage{fullpage}
\usepackage{algorithm}
\usepackage{hyperref}
\usepackage{amsmath, amsthm, amssymb, mathtools, color, framed}
\usepackage{cleveref}
\usepackage{xcolor}
\usepackage{mathrsfs}
\usepackage{algpseudocode}

\usepackage{hyperref}
\hypersetup{
    colorlinks=true,
    linkcolor=blue,
    filecolor=magenta,      
    urlcolor=cyan,
    pdftitle={Overleaf Example},
    pdfpagemode=FullScreen,
    }
\usepackage{tikz}
\usetikzlibrary{matrix,decorations.pathreplacing,calligraphy,calc}
\usepackage{makecell}
\allowdisplaybreaks

\newtheorem{theorem}{Theorem}[section]
\newtheorem{corollary}[theorem]{Corollary}
\newtheorem{remark}[theorem]{Remark}

\newtheorem{claim}[theorem]{Claim}

\newtheorem{observation}[theorem]{Observation}
\newtheorem{definition}[theorem]{Definition}

\newcommand{\bra}[1]{\left\{#1\right\}}
\newcommand{\cbra}[1]{\left\{#1\right\}}
\newcommand{\rbra}[1]{\left(#1\right)}
\newcommand{\mathify}[1]{\ifmmode{#1}\else\mbox{$#1$}\fi}
\DeclarePairedDelimiter\abs{\lvert}{\rvert}

\makeatletter
\let\oldabs\abs
\def\abs{\@ifstar{\oldabs}{\oldabs*}}
\let\oldnorm\norm
\def\norm{\@ifstar{\oldnorm}{\oldnorm*}}
\makeatother

\newcommand{\zone}{\{0, 1\}}

\newcommand{\PDT}{\mathsf{PDT}}
\newcommand{\NADT}{\mathsf{NADT}}
\newcommand{\NAPDT}{\mathsf{NAPDT}}
\newcommand{\PDTsize}{\mathsf{PDTsize}}
\newcommand{\DT}{\mathsf{DT}}
\newcommand{\DTsize}{\mathsf{DTsize}}
\newcommand{\sDT}{\mathsf{sDT}}

\newcommand{\IP}{\mathsf{IP}}
\newcommand{\IND}{\mathsf{IND}}
\newcommand{\MAJ}{\mathsf{MAJ}}

\newcommand{\C}{\mathcal{C}}
\newcommand{\R}{\mathcal{R}}

\newcommand{\F}{\mathbb{F}}

\renewcommand{\epsilon}{\varepsilon}
\newcommand{\bin}{\mathsf{bin}}
\newcommand{\ip}[2]{\langle #1,#2 \rangle}
\newcommand{\eqq}{\stackrel{?}{=}}

\newcommand{\reslinftwo}{$Res(\oplus)$}

\newcommand{\rowreduce}{\textsf{ROW-REDUCE}}

\newcommand{\val}{\mathsf{val}}
\newcommand{\num}{\mathsf{num}}
\newcommand{\FINDVARS}{\mathsf{FINDVARS}}
\newcommand{\PropP}{\mathsf{STIFLE}}
\newcommand{\MARK}{\mathsf{MARK}}
\newcommand{\FREE}{\mathsf{FREE}}
\newcommand{\child}{\textnormal{child}}
\newcommand{\markedindex}{\textsf{markedindex}}
\newcommand{\set}{\mathsf{set}}
\newcommand{\unset}{\mathsf{unset}}

\DeclareMathOperator*{\argmin}{arg\,min}

\newcommand{\size}{\textnormal{size}}

\title{Lifting to Parity Decision Trees Via Stifling}
\author{Arkadev Chattopadhyay\thanks{Tata Institute of Fundamental Research, Mumbai. Partially supported by the MATRICS grant MTR/2019/001633 of the Science \& Engineering Research Board of the DST, India {\tt arkadev.c@tifr.res.in}}
\and
Nikhil S.~Mande\thanks{QuSoft and CWI, Amsterdam. Supported by the Dutch Research Council (NWO), as part of the Quantum Software Consortium programme (project number 024.003.037) {\tt Nikhil.Mande@cwi.nl}}
\and 
Swagato Sanyal\thanks{Indian Institute of Technology, Kharagpur. Supported by an ISIRD grant by SRIC, IIT Kharagpur {\tt swagato@cse.iitkgp.ac.in}}
\and
Suhail Sherif\thanks{Vector Institute, Toronto {\tt suhail.sherif@gmail.com}}}
\date{}

\begin{document}

\maketitle
\begin{abstract}
We show that the deterministic decision tree complexity of a (partial) function or relation $f$ lifts to the deterministic parity decision tree (PDT) \emph{size} complexity of the composed function/relation $f \circ g$ as long as the gadget $g$ satisfies a property that we call \emph{stifling}. We observe that several simple gadgets of constant size, like Indexing on 3 input bits, Inner Product on 4 input bits, Majority on 3 input bits and random functions, satisfy this property. It can be shown that existing randomized communication lifting theorems ([G{\"{o}}{\"{o}}s, Pitassi, Watson.~SICOMP'20], [Chattopadhyay et al.~SICOMP'21]) imply PDT-size lifting. However there are two shortcomings of this approach: first they lift \emph{randomized} decision tree complexity of $f$, which could be exponentially smaller than its deterministic counterpart  when either $f$ is a partial function or even a total search problem. Second, the size of the gadgets in such lifting theorems are as large as logarithmic in the size of the input to $f$. Reducing the gadget size to a constant is an important open problem at the frontier of current research.

Our result shows that even a random constant-size gadget does enable lifting to PDT size. Further, it also yields the first systematic way of turning lower bounds on the \emph{width} of tree-like resolution proofs of the unsatisfiability of constant-width CNF formulas  to lower bounds on the \emph{size} of tree-like proofs in the resolution with parity system, i.e.,~\emph{Res}($\oplus$), of the unsatisfiability of closely related constant-width CNF formulas.  
\end{abstract}

\section{Introduction}\label{sec: intro}
Theorems that lift the complexity of a function in a weaker model of computation to that of the complexity of a closely related function in a stronger model, have proved to be extremely useful and is a dominant theme of current research (see, for example,~\cite{GPW18,CKLM19,GPW20,CFKMP21,LMMPZ22}). One early use of this theme is in the celebrated work of Raz and McKenzie~\cite{RM99}, that built upon the earlier work of Edmonds et al.~\cite{ERIS91} to yield a separation of the hierarchy of monotone circuit complexity classes within NC. The central tool of that work is an argument lifting the query complexity of a (partial) relation $f$ to the two-party communication complexity of the composed relation $f \circ g$, where $g$ is a two-party function, now commonly called a gadget. It required much technical innovation to prove the natural intuition that, if the gadget $g$ is obfuscating enough, the best that the two players can do to evaluate $f \circ g$ is to follow an optimal decision tree algorithm $\mathcal{Q}$ for $f$, solving the relevant instance of $g$ to resolve each query of $\mathcal{Q}$ encountered along its path of execution. The major utility of this theorem lies in the fact that it reduces the difficult task of proving communication complexity lower bounds to the much simpler task of proving decision tree complexity lower bounds. Indeed, this point was driven home more recently, in the beautiful work of G{\"{o}}{\"{o}}s, Pitassi and Watson \cite{GPW18}, who established tight quadratic gaps between communication complexity and (rectangular) partition number, resolving a longstanding open problem. This was, arguably, largely possible as the task was reduced to finding an appropriate analog of that separation in the query world. This last work triggered a whole lot of diverse work on such lifting theorems. Such theorems have been proved for the randomized model \cite{GPW20,CFKMP21}, models with application to data-structures \cite{CKLM18}, proof and circuit complexity (see, for example, \cite{RNV16,GGKS20,GKMP20}), quantum computing \cite{ABK21}, extended formulations \cite{KMR17} etc.

Most of these lifting theorems work with gadgets whose size is at least logarithmic in the arity of the outer function. A current challenge in the area is to break this barrier and prove lifting with sub-logarithmic gadget size (see, for example, \cite{LMMPZ22}), ultimately culminating in constant-size gadgets. It is not clear whether existing ideas/methods can be pushed to achieve this. Faced with this, we consider lifting decision-tree (DT) complexity to intermediate models that are more complex than DT but simpler than 2-party communication. Parity decision tree (PDT) is a natural choice for such a model. Indeed, in tackling another major open problem in communication complexity, the Log-Rank Conjecture (LRC), researchers have realized the importance and utility of such intermediate models. For instance, the natural analogue of LRC is open for PDTs (see \cite{TWXZ13,TXZ16}) giving rise to a basic conjecture in Fourier analysis. The analogue of LRC was recently proved for AND-decision trees~\cite{KLMY21} bringing forth interesting techniques. Further, the randomized analog of LRC, known as the Log-Approximate-Rank Conjecture (LARC) was recently disproved \cite{CMS20} using a surprisingly simple counter-example. The key to discovering this counter-example lay in finding the corresponding counter-example against the analog of the LARC for PDTs. 

Taking cue from these developments, we consider lifting DT complexity using small gadgets to PDT and allied complexity in this work. We are able to find an interesting property of gadgets, which we call `stifling', that we prove is sufficient for enabling such lifting even with constant gadget size.
We define a Boolean function $g : \zone^m \to \zone$ to be \emph{$k$-stifled} if for all subsets of $k$ input variables and for all $b \in \zone$, there exists a setting of the remaining $m - k$ variables that forces the value of $g$ to $b$ (i.e., the values of the $k$ variables are irrelevant under the fixing of these $m - k$ variables). Formally,

\begin{definition}\label{def: stifled}
Let $g : \zone^m \to \zone$ be a Boolean function. We say that $g$ \emph{is $k$-stifled} if the following holds:
\begin{align*}
&\forall S \subseteq [m]~\textnormal{with}~|S| \leq k~\textnormal{and}~\forall b \in \zone,\\ &\exists~z \in \zone^{[m] \setminus S}~\textnormal{such that for all}~x \in \zone^m ~\textnormal{with}~ x|_{[m] \setminus S} = z, g(x) = b.
\end{align*}
\end{definition}

We refer the reader to \Cref{sec: notation} for formal definitions of complexity measures. 
We obtain a lifting theorem from decision tree complexity ($\DT(\cdot)$) of a function $f$ to parity decision tree size ($\PDTsize(\cdot)$, denoting the minimum number of nodes in a parity decision tree computing the function) of $f$ composed with a gadget $g$ that is stifled.\footnote{We remark here that even a lifting theorem from $\DT(f)$ to deterministic communication complexity of $f \circ g$ would not imply \Cref{thm: pdt size lifting} as a black box. This is because deterministic communication complexity can be much larger than the logarithm of PDTsize (for instance, for the Equality function). However a lifting theorem from randomized decision tree complexity of $f$ to randomized communication complexity of $f \circ g$ does imply a lifting theorem from \emph{randomized} decision tree complexity of $f$ to PDTsize of $f \circ g$ (see Claim~\ref{claim: randomized lifting implies dt-pdtsize lifting} and the following discussion).}
\begin{theorem}\label{thm: pdt size lifting}
Let $g:\{0,1\}^m \to \{0,1\}$ be a $k$-stifled function and let $f \subseteq \zone^n \times \R$ be a relation. Then,
\[
\PDTsize(f \circ g) \geq 2^{\DT(f) \cdot k}.
\]
\end{theorem}
In order to prove this, we start with a PDT for $f \circ g$ of size $2^d$, say, and construct a depth-$d/k$ DT for $f$ by `simulating' the PDT.
Using a similar simulation, we also give lower bounds on the \emph{subspace} decision tree complexity of $f \circ g$, where subspace decision trees are decision trees whose nodes can query indicator functions of arbitrary affine subspaces. Let $\sDT(\cdot)$ denote subspace decision tree complexity.

\begin{theorem}\label{thm: sdt depth lifting}
Let $g:\{0,1\}^m \to \{0,1\}$ be a $k$-stifled function, and let $f \subseteq \zone^n \times \R$ be a relation. Then,
\[
\sDT(f \circ g) \geq \DT(f) \cdot k.
\]
\end{theorem}
We observe in \Cref{claim: sdt lifting implies depth-size lifting} that the bounds in the above two theorems are equivalent up to logarithmic factors. More precisely we show that \Cref{thm: pdt size lifting} follows (up to a constant factor in the exponent in the RHS) from \Cref{thm: sdt depth lifting}, and that \Cref{thm: sdt depth lifting} follows (up to a $\log n$ factor in the RHS) from \Cref{thm: pdt size lifting}.
It is easy to see that both of the above theorems individually imply a lifting theorem from $\DT(f)$ to $\PDT(f \circ g)$. We choose to state this as a separate theorem since it is interesting in its own right, and we feel that the proof of \Cref{thm: pdt depth lifting} gives more intuition.

\begin{theorem}\label{thm: pdt depth lifting}
Let $g:\{0,1\}^m \to \{0,1\}$ be a $k$-stifled function, and let $f \subseteq \zone^n \times \R$ be a relation. Then,
\[
\PDT(f \circ g) \geq \DT(f) \cdot k.
\]
\end{theorem}

Examples of stifled gadgets (functions that are $k$-stifled for some integer $k \geq 1$) are the well-studied Indexing function and the Inner Product Modulo 2 function, as we show in \Cref{claim: ind property p} and \Cref{claim: ip stifled}. It is worth noting that we obtain tight lifting theorems for \emph{constant-sized} gadgets. Although there is no gadget of arity $1$ or $2$ that is stifled, both the Indexing gadget on three bits and the Majority gadget on three bits are stifled. Define the Majority function on $n$ input bits by $\MAJ_n(x) = 1$ iff $|\cbra{i \in [n]: x_i=1}| \geq n/2$.
Define the Indexing function as follows.
\begin{definition}[Indexing Function]
For a positive integer $m$, define the Indexing function, denoted $\IND_m : \zone^{m + 2^m} \to \zone$, by
\[
\IND_m(x, y) = y_{\bin(x)},
\]
where $\bin(x)$ denotes the integer in $[2^m]$ represented by the binary expansion $x$.
\end{definition}
In the above definition, we refer to $\cbra{x_i : i \in [m]}$ as the \emph{addressing variables}, and $\cbra{y_j : j \in [2^m]}$ as the \emph{target variables}.

\begin{definition}[Inner Product Modulo 2 Function]
    For a positive integer $m$, define the Inner Product Modulo 2 Function, denoted $\IP_m : \zone^{2m} \to \zone$, by
    \[
    \IP_m(x_1, \dots, x_n, y_1, \dots, y_n) = \oplus_{i = 1}^n (x_i \wedge y_i).
    \]
\end{definition}

We show that $\IND_m$ is $m$-stifled for all integers $m \geq 1$, and $\IP_m$ is $1$-stifled for all integers $m \geq 2$. We refer the reader to \Cref{app: stifled proofs} for proofs of the three claims below. We also show in \Cref{app: stifled proofs} that the Majority function $\MAJ_m$ on $m$ input bits is $(\lceil m/2 \rceil - 1)$-stifled and no Boolean function on $m$ input bits is $\lceil m/2 \rceil$-stifled (hence, `Majority is stiflest').
\begin{claim}\label{claim: ind property p}
For all integers $m \geq 1$, the function $\IND_m$ is $m$-stifled.
\end{claim}

\begin{claim}\label{claim: ip stifled}
For all integers $m \geq 2$, the function $\IP_m$ is $1$-stifled.
\end{claim}

We also show that a random Boolean function (the output on each input is chosen independently to be 0 with probability 1/2 and 1 with probability 1/2) on $m$ input variables (where $m$ is sufficiently large) is $(\log m/2)$-stifled with high probability.
\begin{claim}\label{claim: random stifled}
Let $g : \zone^m \to \zone$ be a random Boolean function with $m$ sufficiently large. Then with probability at least $9/10$, $g$ is $((\log m)/2)$-stifled.
\end{claim}

\Cref{thm: pdt depth lifting},~\Cref{thm: sdt depth lifting} and~\Cref{thm: pdt size lifting} imply the following results for the Indexing gadget in light of \Cref{claim: ind property p}.
\begin{corollary}\label{cor: indexing lifting}
Let $f \subseteq \zone^n \times \R$ be a relation. Then for all positive integers $m$,
\begin{align*}
\PDT(f \circ \IND_m) & \geq \DT(f) \cdot m\\
\sDT(f \circ \IND_m) & \geq \DT(f) \cdot m,\\
\PDTsize(f \circ \IND_m) & \geq 2^{\DT(f) \cdot m}.
\end{align*}
\end{corollary}
We now show that the Parity function witnesses tightness of the above corollary. It is easy to see that for all relations $f$,
\[
\PDT(f \circ \IND_m) \leq \DT(f) \cdot (m + 1).
\]
In fact, our proof can be modified to show $\PDT(f \circ \IND_m) \geq \DT(f) \cdot m + 1$ (see \Cref{rmk: improved simulation}). Let $\oplus_n$ denote the Parity function on $n$ input variables. It outputs 1 if the number of 1s in the input is odd, and 0 otherwise. It is well known that $\DT(\oplus_n) = n$.
We have $\PDT(\oplus_n \circ \IND_m) \leq \DT(\oplus_n) \cdot m + 1 = nm + 1$: query all the $m$ addressing variables in each block to find out the relevant variables using $nm$ queries. The output of the function (which is the parity of all relevant target variables) can now be computed using a single parity query. Thus our result is tight.

Since $\IP_m$ is $1$-stifled for all integers $m \geq 2$ (\Cref{claim: ip stifled}), we analogously obtain the following results for the Inner Product Modulo 2 gadget.
\begin{corollary}\label{cor: ip lifting}
Let $f \subseteq \zone^n \times \R$ be a relation. Then for all positive integers $m \geq 2$,
\begin{align*}
\PDT(f \circ \IP_m) & \geq \DT(f),\\
\sDT(f \circ \IP_m) & \geq \DT(f),\\
\PDTsize(f \circ \IP_m) & \geq 2^{\DT(f)}.
\end{align*}
\end{corollary}
Since a random gadget on $m$ inputs is $(\Omega(\log m))$-stifled (Claim~\ref{claim: random stifled}), we obtain analogous results for random gadgets as well.
\subsection{Consequence for proof complexity}

A central question in the area of proof complexity is ``Given an unsatisfiable CNF $\C$, what is the size of the smallest polynomial-time verifiable proof that $\C$ is unsatisfiable?'' To give lower bounds on the size of such proofs for explicit formulas $C$, the question has been studied with the restriction that the proof must come from a simple class of proofs. Two of these classes are relevant to us: resolution and linear resolution. Resolution is now well understood and strong lower bounds on the size of resolution based proofs are known for several basic and natural formulas, like the pigeonhole principle, Tseitin formulas and random constant-width CNFs. However, linear resolution introduced by Raz and Tzameret \cite{RT08}, where the clauses are conjunctions of affine forms over $\mathbb{F}_2$, are poorly understood. In fact, it remains a tantalizing challenge to prove super-polynomial lower bounds on the size of linear resolution proofs for explicit CNFs. One promising way to make progress on this question is to prove a lifting theorem that lifts ordinary resolution proof-size lower bounds for a CNF to linear resolution proof-size lower bounds for a lifted CNF. Indeed, relatively recently, riding on the success and popularity of lifting theorems on communication protocols, Garg, G{\"{o}}{\"{o}}s, Kamath and Sokolov \cite{GGKS20} showed that ordinary resolution proof-size complexity, among other things, can be lifted to cutting plane proof-size complexity. However, lifting to linear resolution proof-size still evades researchers.

A first step towards the above would be to lift \emph{tree-like} resolution proof-size bounds to tree-like linear resolution proof-size bounds. Although lower bounds for tree-like linear resolution proof-size were established by Itsykson and Sokolov \cite{IS20}, there was no systematic technique known that would be comparable to the generality of a lifting theorem. Our lifting theorem provides the first such technique as explained below.

The starting point is the well-known fact that the size of the smallest tree-like resolution proof that $\C$ is unsatisfiable is the same as the size of the smallest decision tree that computes the relation $S_{\C} := \cbra{(x,C):C\text{ is a clause of }\C\text{ and }C(x)=0}$. Similarly, the size of the smallest tree-like \reslinftwo~(linear resolution over $\F_2$) proof that $\C$ is unsatisfiable is the same as the size of the smallest parity decision tree that computes $S_{\C}$.

Let $\C$ be a CNF with $n$ variables and $t$ clauses, each of width at most $w$. For a gadget $g: \zone^m \to \zone$, we define the CNF $\C \circ g$ as follows. For each clause $C \in \C$, take the clauses of a canonical CNF computing $\Gamma_C : (y_1,\dots,y_n) \in (\zone^m)^n \mapsto C(g(y_1),g(y_2),\dots,g(y_n))$. Note that $\Gamma_C$ is a function of at most $mw$ variables, and thus can be expressed as a CNF with at most $2^{mw}$ clauses, each of width at most $mw$. Define $\C \circ g$ to be the CNF whose clauses are $\bigcup_{C \in \C} \Gamma_C$.

Our results in this paper yield PDT lower bounds against $S_{\C} \circ g$ given DT lower bounds against $S_{\C}$. To compute $S_{\C} \circ g$ we get as input $y = (y_1,\dots,y_n) \in (\zone^m)^n$ and we have to output a clause $C \in \C$ such that $C(x)=0$ where $x = (g(y_1),\dots,g(y_n))$. On the other hand to compute $S_{\C \circ g}$ we get as input $y$ and we have to output a clause $D \in \bigcup_{C \in \C} \Gamma_C$ such that $D(y)=0$. However if $D(y) = 0$ and $D \in \Gamma_C$, then by definition of $\Gamma_C$, $C(x)=0$ where $x = (g(y_1),\dots,g(y_n))$. Hence computing $S_{\C \circ g}$ is sufficient to compute $S_{\C} \circ g$ and our lower bounds against $S_{\C} \circ g$ translate to lower bounds against $S_{\C \circ g}$, bringing us back to the realm of proof complexity.

Our size lifting (\Cref{thm: pdt size lifting}) is particularly relevant, when applied with $g$ as the Indexing gadget on 3 bits, for example. Let $\C$ be a 3-CNF $\C$ with $n$ variables and $t$ clauses, with tree-like resolution width at least $d$. Our theorem implies a tree-like \reslinftwo~size lower bound of $2^d$ for the CNF $\C \circ \IND_1$, with $3n$ variables and $O(t)$ clauses, each of maximum width at most $9$.

\subsection{Related work}
Independently and concurrently with our work, Beame and Koroth~\cite{BK22} obtained closely related results. They show that a particular property of a gadget, conjectured recently by Lovett et al.~\cite{LMMPZ22} to be sufficient for lifting using constant size, is in fact insufficient to yield constant-size lifting. Using properties special to Indexing, they obtain a lifting theorem from decision tree height complexity of $f$ to the complexity of $f \circ \IND$ in a restricted communication model that they define, where $\IND$ has constant size. Further, modifying the proof of this result they also lift decision tree complexity of $f$ to PDT size complexity of $f \circ \IND$.

We, on the other hand, prove directly a lifting theorem for PDT size, from deterministic decision tree height, that works for any constant size gadget that satisfies the abstract property of stifling. Examples of such gadgets are $\IND$, $\IP$, Majority and random functions. While there are some similarities between our and their arguments, especially in the use of row-reduction, our proof is entirely linear algebraic with no recourse to information theory.
\section{The simulation}

In \Cref{sec: notation} we introduce necessary notation for proving our theorems. In \Cref{sec: overview simulation} we give an overview of the simulation for \Cref{thm: pdt depth lifting} (i.e., we show an algorithm that takes a low-depth PDT computing $f \circ g$ and constructs a low-depth DT computing $f$). We do this for the sake of simplicity; the simulations used to prove \Cref{thm: pdt size lifting} and \Cref{thm: sdt depth lifting} are small modifications of the simulation used to prove \Cref{thm: pdt depth lifting}. In \Cref{sec: simulation alg} we show the simulation in full detail, with a key subroutine explained in \Cref{sec: rowreduce}. We then analyze correctness of the simulation in \Cref{sec: correctness}. We conclude the proofs of all of our theorems in \Cref{sec: all simulation theorems}.

\subsection{Notation}\label{sec: notation}
Throughout this paper, we identify $\zone$ with $\F_2$. We identify vectors in $\F_2^n$ with their characteristic set, i.e., a vector $x \in \F_2^n$ is identified with the set $\cbra{i \in [n] : x_i = 1}$. For vectors $a, b \in \F_2^n$, let $\langle a, b \rangle$ denote $\sum_{i = 1}^n a_ib_i$ modulo 2. We say a vector in $\F_2^n$ is non-zero if it does not equal $0^n$. For a positive integer $n$ we use the notation $[n]$ to denote the set $\cbra{1, 2, \dots, n}$ and the notation $[n]_0$ to denote the set $\cbra{0, 1, \dots, n-1}$. Throughout this paper $\R$ denotes an arbitrary finite set. When clear from context, `$+$' denotes addition modulo 2.

\subsubsection{Decision trees and complexity measures}

A \emph{parity decision tree} (PDT) is a binary tree whose leaf nodes are labeled in $\R$, each internal node is labeled by a parity $P \in \F_2^n$ and has two outgoing edges, labeled $0$ and $1$.
On an input $x \in \zone^n$, the tree's computation proceeds from the root down as follows: compute $\val(P,x) := \ip{P}{x}$ as indicated by the node's label and follow the edge indicated by the computed value. Continue in a similar fashion until reaching a leaf, at which point the value of the leaf is output. When the computation reaches a particular internal node, the PDT is said to \emph{query} the parity label of that node.
A decision tree (subspace decision tree, respectively), denoted DT (sDT, respectively), is defined as above, except that queries are to individual bits (indicator functions of affine subspaces of $\F_2^n$, respectively) instead of parities as in PDTs.

A DT/PDT/sDT is said to compute a relation $f \subseteq \zone^n \times \R$ if the output $r$ of the DT/PDT/sDT on input $x$ satisfies $(x, r) \in f$ for all $x \in \zone^n$.
The DT/PDT/sDT complexity of $f$, denoted $\DT(f)/\PDT(f)/\sDT(f)$, is defined as
\[
\DT(f)/\PDT(f)/\sDT(f) := \min_{T : T~\text{is a DT/PDT/sDT computing}~f} \textnormal{depth}(T).
\]
The decision tree size (parity decision tree size, respectively) of $f$, denoted $\DTsize(f)$ ($\PDTsize(f)$, respectively), is defined as
\[
\DTsize(f)/\PDTsize(f) := \min_{T : T~\text{is a DT/PDT computing}~f} \textnormal{size}(T),
\]
where the size of a tree is the number of leaves in it.

\subsubsection{Composed relations and simulation terminology}

For positive integers $n,m$, a relation $f \subseteq \zone^n \times \R$, a function $g : \zone^{m} \to \zone$ and strings $\cbra{y_i \in \zone^{m} : i \in [n]}$, we use the notation $g^n(y_1, \dots, y_n)$ to denote the $n$-bit string $(g(y_1), \dots, g(y_n))$. The relation $f \circ g \subseteq (\zone^{m})^{n} \times \R$ is defined as $(y_1, \dots, y_n, r) \in (f \circ g) \iff (g(y_1), \dots, g(y_n), r) \in f$.

We view the input to $f \circ g$ as an $mn$-bit string, with the bits naturally split into $n$ \emph{blocks}. A parity query $P$ to an input of $f \circ g$ can be specified as an element of $(\F_2^{m})^{n}$. For $i \in [n]$ and a parity query $P$, $P|_i \in \F_2^{m}$ refers to the restriction of $P$ to the $i$th block. For a parity query $P$ and a set of indices $S \subseteq [m] \times [n]$, $P|_S$ refers to the restriction of $P$ to the set of indices $S$.
We say parity $P$ \emph{touches block $i$} or $P$ \emph{touches a set $S$} if $P|_i$ is non-zero or $P|_S$ is non-zero, respectively.

A string $y' \in (\zone^{m})^n$ is said to be a \emph{completion} of a partial assignment $y \in (\{0,1,*\}^{m})^{n}$ if $y'_i = y_i$ for all $i \in [m] \times [n]$ with $y_i \neq *$.

For a node $N$ in a PDT and bit $a \in \zone$, we use $\child(N, a)$ to denote the node reached on answering the parity at the node $N$ as $a$.
For a parity $P$ and a partial assignment $y$ such that $y_i \neq *$ for all $i \in P$, $\val(P,y)$ denotes $\ip{P}{y}$.

\subsection{Overview of the simulation}\label{sec: overview simulation}

In this subsection we sketch and describe our simulation algorithm with the goal of proving \Cref{thm: pdt depth lifting}. The proofs of \Cref{thm: pdt size lifting} and \Cref{thm: sdt depth lifting} follow along similar lines, and are formally proved in \Cref{sec: all simulation theorems}.
We prove Theorem~\ref{thm: pdt depth lifting} via a simulation argument: given a PDT of depth $d$ for $f \circ g$, we construct a DT of depth at most $d/k$ for $f$.

On input $x \in \zone^n$, our decision tree algorithm simulates the PDT traversing a path from its root to a leaf, and outputs the value that the PDT outputs at that leaf. During this traversal the decision tree algorithm queries bits of $x$. It also keeps track of a partial assignment $y \in (\{0,1,*\}^{m})^{n}$ satisfying the following property that guarantees correctness: For all $w \in \zone^n$ consistent with the queried bits of $x$ there is a completion $y'$ of $y$ such that
\begin{itemize}
    \item $g^n(y') = w$, and
    \item the path traversed by the PDT on input $y'$ is the same as that traversed in the simulation.
\end{itemize}
This ensures that once the simulation reaches a leaf of the PDT, every input $w$ that is consistent with the queried bits of $x$ has the same output as the output at the leaf. That is, the simulation is a valid query algorithm for $f$.

To carry out the simulation we assign to every parity query along the path a block that the parity query touches, unless its output value is already determined by $y$. We say that the parity query has that block \emph{marked}. At a high level, if a block has been marked by only a few parity queries, then we have enough freedom in the block to do the following:
\begin{itemize}
    \item set its output (i.e., the value of $g$ on the variables in the block) to any bit we wish by setting some of the variables in the block (this uses the stifling property of $g$), and
    \item retroactively complete this partial assignment so that each of the parity queries marking the block evaluate to any bit of our choosing on the completed input. \end{itemize}
We use these to prove the property that guarantees correctness. The first point helps us in getting a partial assignment $y$ such that $g^n(z) = w$ for all completions $z$ of $y$, and the second helps us complete it to a $y'$ that reaches the leaf we want it to reach. A formal statement of the required properties is in \Cref{claim: correctness}.

During the simulation once a block, say block $i$, has been marked by $k$ parity queries we carefully choose $k$ bits of the block that we will keep unset in our partial assignment throughout the simulation. These $k$ bits are what allow us to perform the completion as in the second bullet above. We then query the value of $x_i$ and set the other $m-k$ bits in block $i$ to ensure that the output of block $i$ is set to $x_i$. The fact that we can do this crucially uses the assumption that the gadget $g$ is $k$-stifled. In the process above, every parity in the PDT simulation marks at most one block, and a variable $x_i$ is not queried until the $i$th block is marked $k$ times. Thus the depth of the resultant DT is at most $d/k$.

A situation that may arise is when a parity $P$ marks block $i$ but a previous parity query $P'$ that marked block $i$ has $P|_i = P'|_i$, for example. However, at some point we may be required to set $P|_i$ to $0$ and $P'|_i$ to $1$ in order to follow the path in the simulation, which is impossible.
To avoid such issues we preprocess each parity query $P$ to ensure that when it marks block $i$, $P|_i$ is not in the linear span of $\cbra{P'|_i : P' \text{ is an earlier parity that marks block }i}$. The simulation algorithm is given in \Cref{alg:pdt to dt} with the preprocessing subroutine given in \Cref{alg:rowreduce}. \Cref{alg:pdt to dt} uses two other functions, described below.

\begin{definition}\label{def: propp fn}
    Let $g : \zone^m \to \zone$ be a $k$-stifled Boolean function. We define the function $\PropP_g$ to be a canonical function that takes as input a set $S \subseteq [m]$ with $\abs{S} \leq k$ and a bit $b$ and outputs a partial assignment $y \in \{0,1,*\}^m$ such that
    \begin{itemize}
        \item for all $i \in [m]$, $i \in S \iff y_i = *$, and
        \item for all completions $y'$ of $y$, $g(y') = b$.
    \end{itemize}
\end{definition}
The existence of such a function in the definition above is guaranteed since $g$ is $k$-stifled.

\begin{definition}\label{def: findvars}
Define the function $\FINDVARS$ to be a canonical function that takes as input a set $\cbra{P_1,\dots,P_{\ell}}$ of $\ell$ linearly independent vectors in $\F_2^{m}$ and outputs a set $S$ of $\ell$ indices in $[m]$ such that $\cbra{P_i|_S}_{i \in [\ell]}$ are also linearly independent vectors.
\end{definition}
The existence of such a function in the definition above is guaranteed by the fact that an $\ell \times m$ matrix of rank $\ell$ has $\ell$ linearly independent columns, and the $\ell \times \ell$ matrix defined by restricting the original matrix to these columns also has rank $\ell$.

\subsection{The simulation algorithm}\label{sec: simulation alg}

In Algorithm~\ref{alg:pdt to dt}, we start with a PDT $T$ of size $2^{d}$ for $f \circ g$, and an unknown input $x \in \zone^n$. Our simulation constructs a decision tree for $f$ of depth at most $d/k$.\footnote{Our algorithm is designed to work with $k \geq 1$. For instance, Line~\ref{line: markandkth} is intended to catch the $k$th time a block is marked, which does not make sense for $k=0$.\label{footnote: kgreaterthanzero}} This would prove \Cref{thm: pdt size lifting}.
We use the following notation in Algorithm~\ref{alg:pdt to dt}:
\begin{itemize}
    \item $\num$ is the number of parity queries simulated in $T$ so far.
    \item $M$ is an array storing parity queries made so far (after processing).
    \item $A$ is an array storing the answers to the parities in $M$. These are such that answering the queries in $M$ with the answers in $A$ is equivalent to answering the queries in the PDT to reach the node reached by the simulation algorithm.
    \item For all $i \in [n]$, $\MARK[i]$ is the set of indices in $M$ of parity queries that have the $i$th block marked. That is, $\MARK[i] = \cbra{j \in [\num]_0 : M[j] \text{ marks block }i}$.
    \item $y$ is the partial assignment stored by the simulation algorithm. It is initialized to $(*^m)^n$.
    \item For all $i \in [n]$, $\FREE[i]$ is the set of indices in block $i$ of $y$ that have not yet been set. That is, $\FREE[i] = \cbra{(i,j) : j \in [m], y_{i,j} = *}$. It is initialized to $\cbra{i} \times [m]$ for each $i \in [n]$.
    \item $Q$ is the set of indices of $x$ queried during the simulation.
\end{itemize}

The algorithm starts at the root node of the PDT and does the following for each parity query $P$ it comes across in its traversal down the PDT (Line~\ref{line: simulationloop}).
\begin{itemize}
    \item In Line~\ref{line: rowreduce} it processes the parity query using the \rowreduce~subroutine. The subroutine outputs a parity $P'$, a `correction' bit $b$ and an index $j \in [n]$ denoting which block was marked by the parity (or $\bot$ if no block was marked). 
As stated in Claim~\ref{clm: processedequivalence}, these outputs satisfy the following: for all $y \in \zone^{mn}$,
    \[ \Big[\val(M[i], y) = A[i]\, \forall i \in [\num]_0\Big] \implies \Big[\val(P, y) = 0 \iff \val(P', y) = b\Big]. \]
    In other words, for all inputs consistent with the answers to the previous parity queries made along the path, $P$ evaluates to $0$ if and only if $P'$ evaluates to $b$.
    \item It adds $P'$ to the list $M$ of processed parities in Line~\ref{line: mnump'}. There are now two possibilities:
    \begin{itemize}
        \item If $P'$ does not mark a block (Line~\ref{line: nomark}), then 
all variables appearing in $P'$ are already set in $y$ (see \Cref{claim: processedstructure}). Hence $A[\num]$ is set to $\val(P',y)$ (Line~\ref{line: setval}) and the simulation algorithm answers $\val(P',y)+b$ to the original parity query at the current node (Line~\ref{line: setchild}).
        \item If $P'$ marks block $j$:
        \begin{itemize}
            \item If this is the $k$th parity to mark block $j$ (Line~\ref{line: markandkth}), the simulation queries the value of $x_j$ (Line~\ref{line: query}). It then uses the $\FINDVARS$ function to select $k$ bits in the $j$th block to keep unset (Line~\ref{line: freechoice}). It sets the remaining $m-k$ bits of the $j$th block of $y$ using the $\PropP_g$ function to ensure that the output of the $j$th block is $x_j$ (Line~\ref{line: setvars}). The stifling property of $g$ is crucially used in this step.
            \item The algorithm sets the value of $A[\num]$ (in Line~\ref{line: gotosmallerchild}) such that the simulation algorithm answers the parity query at the current node (Line~\ref{line: setchild}) so as to go to the smaller subtree.
        \end{itemize}
    \end{itemize}
\end{itemize}

When it reaches a leaf it outputs the value at that leaf.

\begin{algorithm}
    \caption{Simulation of a PDT $T$ for $f \circ g$ to obtain a DT for $f$}\label{alg:pdt to dt}
    \begin{algorithmic}[1]
    \Require $x \in \zone^n$
    \State $P_0 \gets \textnormal{root}(T)$ \Comment{Starting the simulation of the PDT}
    \State{$\num \gets 0$}
    \State{$M \gets []$}
    \State{$A \gets []$}
    \State{$\MARK[i] \gets \emptyset\,\forall i \in [n]$}
    \State{$y \gets (*^{m})^{n}$}
    \State{$\FREE[i] \gets \cbra{i} \times [m]\,\forall i \in [n]$}
    \State{$Q \gets \emptyset$}
    \While{$P_\num$ is not a leaf of $T$}\label{line: simulationloop}
    \State{$P',b,j \gets$ \rowreduce($P_{\num}, M, A, \MARK, \FREE$)}\label{line: rowreduce}\Comment{Process the current parity query and return the processed query, a correction bit and the index of the marked block.}
    \State{$M[\num] \gets P'$}\label{line: mnump'}
    \If{$j = \bot$}\label{line: nomark}\Comment{If no block has been marked}
        \State{$A[\num] \gets \val(M[\num],y)$}\label{line: setval}\Comment{By Claim~\ref{claim: processedstructure}, all variables appearing in $M[\num]$ have been set in $y$.}
    \Else\label{line: elsemarked}
        \State{$\MARK[j] \gets \MARK[j] \cup \bra{\num}$}\label{line: markjappend}
        \If{$\abs{\MARK[j]} = k$}\label{line: markandkth}\Comment{If this is the $k$th time the $j$th block is marked}
            \State{Query $x_j$}\label{line: query}
            \State{$Q \gets Q \cup \cbra{j}$}\label{line: updateq}\Comment{Update set of queried variables.}
            \State{$S \gets \FINDVARS(\cbra{M[i]|_j : i \in \MARK[j]})$}\label{line: freechoice}\Comment{Select $k$ bits of block $j$ that will remain unset in $y$.}
            \State{$y|_j \gets \PropP_g(S,x_j)$}\label{line: setvars}\Comment{Set other bits in block $j$ to force the $j$th block to output value~$x_j$. This is possible since~$g$ is $k$-stifled.}
            \State{$\FREE[j] \gets \cbra{j} \times S$}\label{line: update free}
\EndIf
        \State{$b' \gets \argmin_{a} \size(\child(P_{\num},a))$}\label{line: b'choice}
        \State{$A[\num] \gets b'+b$}\label{line: gotosmallerchild}
    \EndIf
    \State{$P_{\num+1} \gets \child(P_\num,A[\num]+b)$}\label{line: setchild}\Comment{If no block was marked (Line~\ref{line: nomark}-\ref{line: setval}), go to the child as dictated by the set variables. Else, go to the child with the smaller subtree, as ensured by Lines~\ref{line: b'choice} and~\ref{line: gotosmallerchild}.}
    \State{$\num \gets \num + 1$}
    \EndWhile
    \State{Output the value output at $P_\num$}
    \end{algorithmic}
\end{algorithm}

\subsection{The \rowreduce~subroutine}\label{sec: rowreduce}

We now discuss the \rowreduce~subroutine, which describes how to process an incoming parity and choose which block it should mark.
Algorithm~\ref{alg:rowreduce} takes as inputs $P$,$M$,$A$,$\MARK$,$\FREE$, each of which are described below.
\begin{itemize}
    \item $P$ is the new parity query to be processed.
    \item The previously processed parity queries are stored in $M$.
    \item The answers given in the simulation to the parities in $M$ are stored in $A$.
    \item $\MARK[i]$ is the set of indices of parity queries that have the $i$th block marked.
    \item $\FREE[i]$ is the set of unset variables in block $i$ in the partial assignment (referred to as $y$ in \Cref{sec: overview simulation}) of the input to $f \circ g$ that the simulation keeps.
\end{itemize}

\begin{algorithm}
\caption{The \rowreduce~routine}\label{alg:rowreduce}
\begin{algorithmic}[1]
\Require $P$,$M$,$A$,$\MARK$,$\FREE$

\State{$\markedindex \gets \bot$}
\State{$P' \gets P$}
\State{$b \gets 0$}
\For{$j$ from $1$ to $n$}\label{line: forloop}
\If{$P'|_{\FREE[j]}$ is in the linear span of $\bra{M[i]|_{\FREE[j]}: i \in \MARK[j]}$}\label{line: iflinspan}
    \State{Let $S \subseteq \MARK[j]$ be such that $P'|_{\FREE[j]} = \sum_{i \in S} M[i]|_{\FREE[j]}$.}
    \State{$P' \gets P' + \sum_{i \in S} M[i]$}\label{line: zeroout}\Comment[.65]{Observe that $P'|_{\FREE[j]} = 0^{\FREE[j]}$ after this step.}
    \State{$b \gets b + \sum_{i \in S} A[i]$}\label{line: zerooutb}
\Else\label{line: else}
    \State{$\markedindex \gets j$}\label{line: setmarkedindex}
    \State{\textbf{break}}\label{line: break}\Comment[.65]{Ensure that no more than 1 block is marked.}
\EndIf
\EndFor
\State{\Return{$P',b,\markedindex$}}
\end{algorithmic}
\end{algorithm}

The processed parity $P'$ is initialized to $P$, and a correction bit $b$ is initialized to $0$. In Line~\ref{line: forloop}, we go through the blocks of variables in order looking for a block to assign the incoming parity to.
The check made on Line~\ref{line: iflinspan} is to see whether this parity can mark the current block $j$. It cannot mark block $j$ if its restriction to $\FREE[j]$ is in the linear span of the restrictions of earlier parities that marked block $j$. We detail what the algorithm does based on this check below.
\begin{itemize}
\item In Lines~\ref{line: iflinspan} to~\ref{line: zerooutb}, we consider the case where the parity being processed, restricted to the $j$th block, is a linear combination of the restrictions of earlier parities that marked the $j$th block. That is, there is a set $S \subseteq \MARK[j]$ such that $P'|_{\FREE[j]} = \sum_{i \in S} M[i]|_{\FREE[j]}$. In Line~\ref{line: zeroout} we update the parity by adding $\sum_{i \in S} M[i]$ to $P'$. This ensures that $P'|_{\FREE[j]} =0$. Then in Line~\ref{line: zerooutb} we update the correction bit $b$ by adding $\sum_{i \in S} A[i]$ to it. This is so that the following holds for all $y$.
    \[ \Big[\val(M[i], y) = A[i]\, \forall i \in [\num]_0\Big] \implies \Big[\val(P, y) = 0 \iff \val(P', y) = b\Big]. \]
    \item In Lines~\ref{line: else} to~\ref{line: break}, we consider the case when the current parity restricted to $\FREE[j]$ is not a linear combination of the restrictions of the earlier parities that marked the $j$th block. In this case we have the current parity mark the $j$th block (Line~\ref{line: setmarkedindex}) and then return $P',b,j$.
\end{itemize}
Notice that a block is marked only in the latter case. If this case does not occur for any of the $[n]$ blocks, then no block is marked and the procedure returns $P',b,\bot$.

\subsection{Correctness}\label{sec: correctness}

We first address the fact that we primarily deal with the values of the processed parities rather than the original parity queries from the PDT. The following claim says that answering the processed parity is in fact equivalent to answering the original parity from the PDT.

\begin{claim}\label{clm: processedequivalence}
    Let $P',b,j$ be the output of \rowreduce~when run with its first three inputs being $P,M$ and $A$.  Let $\num$ be the number of entries in $M$. Then for all $y \in \zone^{mn}$,
    \[ \Big[\val(M[i], y) = A[i]\, \forall i \in [\num]_0\Big] \implies \Big[\val(P, y) = 0 \iff \val(P', y) = b\Big]. \]\end{claim}

\begin{proof}
    The output $P'$ is computed from the input $P$ only through modifications by Line~\ref{line: zeroout} of Algorithm~\ref{alg:rowreduce}. That is, there is some subset $S \subseteq [\num]_0$ such that $P' = P + \sum_{i \in S} M[i]$. The output $b$ is also modified alongside the modifications of $P'$ so that $b = \sum_{i \in S} A[i]$.
    
    Hence for any input $y$, $\val(P',y) = \val(P,y) + \sum_{i \in S} \val(M[i],y)$. The claim immediately follows.
\end{proof}

As a corollary, we get that analyzing inputs that answer all the processed parities according to the answer array $A$ is the same as analyzing inputs that follow the path taken by the PDT simulation.

\begin{corollary}\label{cor: processedsimulation}
    The following statement is a loop invariant for the while loop in Line~\ref{line: simulationloop} of Algorithm~\ref{alg:pdt to dt}. For all $y \in \zone^{mn}$,
    \[ \Big[\val(M[i], y) = A[i]\, \forall i \in [\num]_0\Big] \iff \Big[y \text{ reaches }P_{\num}\Big]. \]
\end{corollary}

\begin{proof}
    The loop is initially entered with the array $M$ being empty and $P_{\num}$ being the root of the PDT. Clearly the statement holds at this point. We now prove by induction that it is a loop invariant.
    
    Suppose the statement holds at the beginning of an execution of the loop. $M[\num]$ is populated in this execution by the parity output by the \rowreduce~subroutine in Line~\ref{line: rowreduce}. By Claim~\ref{clm: processedequivalence} and the fact that the statement held at the beginning of the execution, every input $y$ that reaches the node $P_{\num}$ satisfies $\val(M[\num],y) = \val(P_{\num},y) + b$. In Line~\ref{line: setchild} we go to the node $P_{\num+1} = \child(P_{\num},A[\num]+b)$. Hence the inputs that reach $P_{\num+1}$ are exactly those that reach $P_{\num}$ and satisfy $M[\num] = A[\num]$. We then increment $\num$ and the loop ends, so the statement holds at the end of an execution of the loop.
\end{proof}

We now show that the processed parities are very structured.

\begin{claim}\label{claim: processedstructure}
Let $\num$ be the number of queries made by the simulated PDT at a certain point in the simulation. Recall that the processed parities are stored in $M[0]$ to $M[\num-1]$.
\begin{enumerate}
    \item Let $i \in [\num]_0$. If the parity $M[i]$ marks block $j$ (i.e., $i \in \MARK[j]$), then for every block $j_1<j$,
    \[ M[i] \cap \FREE[j_1] = \emptyset \]
    where $M[i]$ is viewed as the set of indices $\cbra{\alpha \in [m] \times [n] : M[i]_{\alpha} = 1}$.\\
    If $M[i]$ does not mark any block, then the above holds for all $j_1 \in [n]$.
    \item For every block $j \in [n]$, the parities $\bra{M[i]|_{\FREE[j]} : i \in \MARK[j]}$ are linearly independent.
\end{enumerate}
\end{claim}

\begin{proof}
    We prove the two properties in order.
    \begin{enumerate}
        \item We prove the first property by induction on $i$. 
        \begin{itemize}
            \item Let's assume that the property holds for all $i'<i$. Note that $M[i]$ is the processed parity output by \rowreduce~when it was called to process $P_i$. Similarly $j$ is the block marked by \rowreduce~during the same call. During the processing, the algorithm ran through all blocks $j_1 < j$ in order (or $j_1 \leq n$ if no block was marked) and took the branch of Line~\ref{line: iflinspan} of Algorithm~\ref{alg:rowreduce} in each. To prove that $M[i] \cap \FREE[j_1] = \emptyset$ for each $j_1 < j$, we use a second induction on $j_1$.
            \begin{itemize}
                \item When the algorithm looks at block $j_1$ (i.e., in the $j_1$th iteration of Line~\ref{line: forloop} in Algorithm~\ref{alg:rowreduce}), we assume that the partially processed parity $P'$ has, for each $j_2<j_1$, $P' \cap \FREE[j_2] = \emptyset$. When looking at block $j_1$, the algorithm may process $P'$ further by adding to it some previously processed parities that had marked block $j_1$ (Line~\ref{line: zeroout}). But any previously processed parity would be stored in $M[i']$ for some $i'<i$. By our first induction hypothesis we know that any such $M[i']$ that had marked block $j_1$ has, for each $j_2<j_1$, $M[i'] \cap \FREE[j_2] = \emptyset$. Hence even after adding such parities to $P'$ it still holds that for each $j_2<j_1$, $P' \cap \FREE[j_2] = \emptyset$. The added parities explicitly ensure that $P' \cap \FREE[j_1] = \emptyset$ and so for each $j_2 \leq j_1$, $P' \cap \FREE[j_2] = \emptyset$. This completes the second induction.
            \end{itemize}
            Since this second induction proves the required statement for all $j_1<j$, this completes the first induction as well.
        \end{itemize}
        \item For the second property we note that if a processed parity $P'$ marks a block $j$ (Line~\ref{line: setmarkedindex} of Algorithm~\ref{alg:rowreduce}), then $P'|_{\FREE[j]}$ is not in the linear span of $P''|_{\FREE[j]}$ for previous parities $P''$ that have block $j$ marked (since the \textbf{if} condition Line~\ref{line: iflinspan} was not satisfied). Hence the set of parities that mark a block are necessarily linearly independent. Moreover, when $\FREE[j]$ is modified (see Line~\ref{line: freechoice} of Algorithm~\ref{alg:pdt to dt}), it is chosen to be the output of the function $\FINDVARS$ on input $\cbra{M[i]|_j : i \in \MARK[j]}$. By the definition of $\FINDVARS$, this ensures that $\cbra{M[i]|_{\FREE[j]} : i \in \MARK[j]}$ remains linearly independent.
    \end{enumerate}
\end{proof}

We finally prove the correctness of the simulation.

\begin{claim}\label{claim: correctness}
For any string $w \in \zone^n$ consistent with the queried bits of $x$ there is an input $y$ to $f \circ g$ that reaches the leaf $P_{\num}$ and satisfies $g^n(y) = w$.
\end{claim}

\begin{proof}

Fix a string $w \in \zone^n$ consistent with the queried bits of $x$. We start with the partial assignment $y$ maintained by the PDT simulation and complete it.

\begin{itemize}
    \item \textbf{Fixing variables in $y$ to ensure $g^n(y) = w$}: 
The bits of $y$ that are already fixed ensure that for any queried bit $x_j$, $g(y|_j) = w_j$ (these are set in Line~\ref{line: setvars} of Algorithm~\ref{alg:pdt to dt}).

Now for each unqueried bit $x_j$, first note that $\ell := \abs{\MARK[j]} < k$ (otherwise Line~\ref{line: query} of Algorithm~\ref{alg:pdt to dt} ensures that $x_i$ is queried). We do the following:
\begin{itemize}
    \item Get a set $S = \FINDVARS(\cbra{M[i]|_j}_{i \in \MARK[j]})$ of $\ell$ indices such that $\cbra{M[i]|_S}_{i \in \MARK[j]}$ are still linearly independent.
    \item Assign $y|_j$ to be $\PropP_g(S,w_j)$ to ensure that all the bits of $y|_j$ indexed in $S$ are unset, those outside are set, and all completions of $y|_j$ evaluate to $w_j$ when $g$ is applied to them. This is possible since $g$ is $k$-stifled and $\ell < k$.\footnote{It suffices to assume $\ell \leq k$ here. See \Cref{rmk: improved simulation} for a discussion on how we can modify our simulation to exploit this.} We will continue to use $\FREE[j]$ to refer to the set $S$.
\end{itemize}

Now we are guaranteed that any extension of $y$ with bits set as above will yield the string $w$ when $g^n$ is applied to it.

Before we move on to ensuring that $y$ reaches the leaf $P_{\num}$, let us make the useful observation that at this point $|\MARK[j]| = |\FREE[j]|$ for all $j \in [n]$. We have ensured this for all $j$ that were unqueried by the simulation. For those that were queried, note that the simulation only queries $x_j$ when $|\MARK[j]|=k$, and on querying $x_j$, $\FREE[j]$ is modified to have size $k$ (Lines~\ref{line: markandkth} and~\ref{line: update free} of Algorithm~\ref{alg:pdt to dt}). By Claim~\ref{claim: processedstructure}, we know that $\bra{M[i]|_{\FREE[j]} : i \in \MARK[j]}$ are linearly independent parities. Hence in Algorithm~\ref{alg:rowreduce}, when processing block $j$, Line~\ref{line: iflinspan} will always fire and block $j$ is never marked again. This ensures that $\MARK[j]$ and $\FREE[j]$ do not get modified again.

\item \textbf{Extending $y$ to guarantee that it reaches the leaf $P_{\num}$}:

By \Cref{cor: processedsimulation} the inputs $y$ that reach the leaf $P_{\num}$ are exactly those that answer $A[i]$ to each processed parity query $M[i]$.

We analyze each parity query in $M$. There are two cases: one where the current parity query under consideration marks a block, and the other where the parity query under consideration does not mark a block.

\begin{itemize}
    \item \textbf{When $i \notin \bigcup_j \MARK[j]$:} By Claim~\ref{claim: processedstructure} the only non-zero entries of $M[i]$ are variables that have previously been set in $y$ (in Line~\ref{line: setvars} of Algorithm~\ref{alg:pdt to dt}). Hence the parity query $M[i]$ has to evaluate to $\val(M[i],y)$ when queried on any extension of $y$, and indeed $A[i]$ is set to $\val(M[i],y)$ in Line~\ref{line: setval} of Algorithm~\ref{alg:pdt to dt}. 

    \item \textbf{When $i \in \bigcup_j \MARK[j]$:}
Let $I = \bigcup_j \MARK[j]$. Note that $|I| = \sum_j \abs{\MARK[j]}$ since at most one block gets marked for each parity. Since $\abs{\FREE[j]} = \abs{\MARK[j]}$ for all $j \in [n]$, there are still $\abs{I}$ unset bits of $y$. Let us refer to these variables by $\unset(y)$ and their complement by $\set(y)$.

\newcommand{\mM}{\mathcal{M}}
\newcommand{\mA}{\mathcal{A}}
\newcommand{\mb}{b}
Let us now construct a matrix $\mM \in \F_2^{I \times [mn]}$ where the rows are indexed by $I$ and row $i$ is the vector $M[i]$. We also consider the column vector $\mb \in \F_2^{I}$ where the rows are indexed by $I$ and the entry $\mb_i = A[i]$.

Our goal is to find a column vector $v \in \F_2^{[mn]}$ such that $v_{\set(y)} = y_{\set(y)}$ and $\mM v = \mb$. Setting the unset bits of $y$ according to $v$ would give us what we want, a completion of $y$ that answers $A[i]$ to each parity query $M[i]$. To this end let us write the matrix $\mM$ as follows (after an appropriate permutation of columns).
\[
\mM = \begin{bmatrix}
\mA &|  & \mA'
\end{bmatrix},
\]
where the first set of columns corresponds to indices in $\unset(y)$, and the second set of columns corresponds to indices in $\set(y)$.

Let $b' = \mA'v_{\set(y)}$. If we find a setting of $v_{\unset(y)}$ such that $\mA v_{\unset(y)} = \mb+b'$, then $\mM v = \mb$ and this would conclude the proof. We prove the existence of such a setting by noting below that the rows of $\mA$ are independent and hence its image is $\F_2^{I}$. To prove the independence of the rows of $\mA$, view $\mA$ as a block matrix made up of submatrices $\bra{\mA|_{i,j}}_{i,j \in [n]}$ with $\mA|_{i,j}$ containing the rows $\MARK[i]$ and columns $\FREE[j]$ (see Figure~\ref{fig: matrix A}). Note that submatrices of the form $\mA|_{i,i}$ are square matrices since $\abs{\MARK[i]} = \abs{\FREE[i]}$ for all $i \in [n]$.

\begin{figure}
    \centering
    \begin{tikzpicture}
        \foreach \i in {1,2,4}
        \foreach \j in {1,2,4}
        {
            \node at (1.5*\i,-1.5*\j) {$\mA|_{\ifthenelse{\j=4}{n}{\j},\ifthenelse{\i=4}{n}{\i}}$};
        }
        \foreach \i in {1,2,4}
        {
            \node at (1.5*3,-1.5*\i) {$\cdots$};
            \node at (-0.5,-1.5*\i) {$\MARK[\ifthenelse{\i=4}{n}{\i}]$};
            \draw [decorate,decoration = {calligraphic brace}] (0.4,-1.5*\i-0.7) --  (0.4,-1.5*\i+0.7);
    
            \node at (1.5*\i,-1.5*3) {$\vdots$};
            \node at (1.5*\i,-0.1) {$\FREE[\ifthenelse{\i=4}{n}{\i}]$};
            \draw [decorate,decoration = {calligraphic brace}] (1.5*\i-0.7,-0.6) --  (1.5*\i+0.7,-0.6);
        }
        \node at (1.5*3,-1.5*3) {$\ddots$};
        \draw[thick] (0.75,-0.75) -- (0.6,-0.75) -- (0.6,-6.75) -- (0.75,-6.75);
        \draw[thick] (6.75,-0.75) -- (6.9,-0.75) -- (6.9,-6.75) -- (6.75,-6.75);
    \end{tikzpicture}
    \caption{The structure of the matrix $\mA$. $\mA|_{i,i}$ is a linearly independent submatrix for each $i \in [n]$, and
    $\mA|_{i,j}$ is the all-$0$ matrix when $i>j$.}
    \label{fig: matrix A}
\end{figure}

By \Cref{claim: processedstructure} and the way we defined $\FREE[i]$ using the $\FINDVARS$ function (ensuring that the restriction of parities in $\MARK[i]$ to $\FREE[i]$ preserves independence),
\begin{itemize}
    \item $\mA|_{i,i}$ is a linearly independent submatrix for each $i \in [n]$, and
    \item $\mA|_{i,j}$ is the all-$0$ matrix when $i>j$.
\end{itemize}
The independence of the rows of $A$ easily follows from this `upper triangular'-like structure.
\end{itemize}
\end{itemize}
\end{proof}

We require the following observation for the proof of \Cref{thm: sdt depth lifting} since we work with a slightly modified simulation algorithm there.
\begin{observation}\label{obs: markedanswersirrelevant}
    The correctness analysis of Algorithm~\ref{alg:pdt to dt} described in this section still holds if we set $A[\num]$ to $b_{\num} \in \zone$ in Line~\ref{line: gotosmallerchild}, where $b_{\num}$ is an arbitrary function of the history of the simulation.
\end{observation}

Finally we observe that the number of queries made by the simulation is at most $d/k$.
\begin{observation}\label{obs: cost}
    Algorithm~\ref{alg:pdt to dt} makes at most $d/k$ queries to $x$.
\end{observation}

\begin{proof}
    First observe that every parity in the simulation marks at most one block.
    The query $x_j$ is made by the simulation (Line~\ref{line: query}) only when $\MARK[j]$ gets $k$ elements in it (Line~\ref{line: markandkth}). So the number of queries made is at most $\abs{\bigcup_{j \in [n]} \MARK[j]}/k \leq d/k$, since the number of iterations in the simulation (captured by $\num$) is at most the depth of the PDT, which is $d$ by assumption.
\end{proof}

It follows from the correctness of \Cref{alg:pdt to dt} and the observation above that $\PDT(f \circ g) \geq \DT(f) \cdot k$ (proving \Cref{thm: pdt depth lifting}). We observe below that our simulation can be modified so as to give a stronger bound of $\PDT(f \circ g) \geq \DT(f) \cdot k + 1$. Note that this bound is tight since $\PDT(\oplus_n \circ \IND_1) \leq n+1$, for instance.
\begin{remark}\label{rmk: improved simulation}
Note that we query $x_j$ in Algorithm~\ref{alg:pdt to dt} when the $j$th block is marked for the $k$th time. The reason behind this is that we can handle a block being marked $k$ times while still having the freedom to set its output by the stifling property of $g$. Moreover this might not be true after the block is marked for the $(k+1)$th time. However querying $x_j$ in \emph{the same iteration} where the $j$th block is marked for the $k$th time (Lines~\ref{line: markandkth} and~\ref{line: query}) is premature. We can modify our algorithm so as to query $x_j$ after block $j$ is marked $k$ times \emph{and} a new parity \emph{attempts to} mark it for the $(k+1)$th time. With this modification, once we query $x_j$ and set variables so that its output is fixed, the parity can no longer mark block $j$ and needs to be reprocessed using Algorithm~\ref{alg:rowreduce} with the updated value of $\FREE[j]$ so that it either triggers a query in a later block (in which case we repeat this process), marks a later block, or marks no block at all. If the set of queries to $x$ is denoted by $Q$, then the number of parity queries being made by the simulation is at least $|Q|k + 1$: When the last of the queries in $Q$ is made, say $x_j$, then the $j$th block and all blocks corresponding to previously queried variables had already been marked $k$ times \emph{before} the parity query appeared that triggered the query of $x_j$. Hence there were at least $|Q|k$ parity queries processed \emph{before} the parity query that resulted in the simulation querying $x_j$. As a result we get 
\[
\PDT(f \circ g) \geq \DT(f) \cdot k + 1.
\]
\end{remark}

\subsection{Simulation theorems}\label{sec: all simulation theorems}
In this section we conclude the proofs of \Cref{thm: pdt depth lifting}, \Cref{thm: sdt depth lifting}, \Cref{thm: pdt size lifting}, \Cref{cor: indexing lifting} and \Cref{cor: ip lifting}.

\begin{proof}[Proof of \Cref{thm: pdt depth lifting}]
    It follows from the correctness of Algorithm~\ref{alg:pdt to dt} described in the previous section, and Observation~\ref{obs: cost}.
\end{proof}

\begin{proof}[Proof of \Cref{thm: pdt size lifting}]
Consider a PDT $T_P$ for $f \circ g$ of size $s$. Suppose the DT $T$ we get for $f$ from the simulation makes $d$ queries on a worst-case input, say on $x \in \zone^n$.
Note that whenever a parity marks a block in the simulation algorithm (i.e., Line~\ref{line: elsemarked} fires), the algorithm then chooses to go to the child with the smaller subtree (Lines~\ref{line: gotosmallerchild} and~\ref{line: setchild}). So in Line~\ref{line: setchild}, the size of the subtree rooted at $P_{\num+1}$ is at most half the size of the subtree rooted at $P_{\num}$.
When the simulation algorithm is run on input $x$, it reaches a leaf of $T_P$, and at least $dk$ parities are such that they mark some block in the process. This is because the simulation queries a variable only after its corresponding block has been marked $k$ times. Hence the size of the subtree rooted at the reached leaf is at most a $2^{-dk}$ fraction of $s$. Thus, $s \geq 2^{dk}$. The theorem now follows from the correctness of Algorithm~\ref{alg:pdt to dt} described in the previous section.
\end{proof}

Towards the proof of \Cref{thm: sdt depth lifting}, we work with a modified query model instead of considering subspace decision trees. This query model is described below.

For each affine subspace choose an arbitrary ordering of the constraints.
We define a `parity constraint' query to be a query of the form $\ip{v}{x}\eqq a$ where $v \in \mathbb{F}_2^{mn}, a \in \mathbb{F}_2$.
Define a `parity constraint' query protocol as a protocol that makes parity constraint queries. The answer to a query $\ip{v}{x}\eqq a$ is either `Right' if $\ip{v}{x}=a$ or `Wrong' if $\ip{v}{x} \neq a$. The cost of the protocol is defined to be the worst-case number of queries that were answered `Wrong' during a run of the protocol.

Now given a subspace decision tree $T$, we construct a parity constraint query protocol as follows:
\begin{itemize}
    \item When an affine subspace query is made in $T$, instead query the parity constraints making up the affine subspace query in a canonical order. Stop when any of the parity constraints is answered `Wrong'.
    \item If a constraint has been answered `Wrong', traverse $T$ as if the affine subspace query had failed. If no constraint has been answered `Wrong', traverse $T$ as if the affine subspace query had succeeded.
\end{itemize}
Clearly the number of wrong answers in any execution of this protocol is upper bounded by the worst case number of affine subspace queries made by $T$, and the parity constraint protocol computes the same relation as the original subspace decision tree.

\begin{proof}[Proof of \Cref{thm: sdt depth lifting}]
From a subspace decision tree computing $f \circ g$, construct a parity constraint query protocol $T$ computing $f \circ g$ without increasing the cost, as described above.
Note that a parity constraint query protocol tree is just a parity decision tree except that for every internal node its outgoing edges have an extra label of `Right' and `Wrong'.

    We modify Lines~\ref{line: b'choice} and~\ref{line: gotosmallerchild} of Algorithm~\ref{alg:pdt to dt} as follows. We look at the current query node in the PDT to see what answer $b' \in \F_2$ corresponds to a `Wrong' edge. We set $A[\num]$ to $b' + b$ so that in Line~\ref{line: setchild} the simulation does take the `Wrong' edge.
    
    Hence for every $i \in \bigcup_{j \in [n]} \MARK[j]$, the parity query $M[i]$ corresponds to a `Wrong' answer. Thus note that cost$(T) \geq \abs{\bigcup_{j \in [n]} \MARK[j]}$. On the other hand, the query $x_j$ is made by the simulation only when $\MARK[j]$ gets $k$ elements in it. So the number of queries made is at most $\abs{\bigcup_{j \in [n]} \MARK[j]}/k \leq \text{cost}(T)/k$. The correctness of this modified algorithm follows from the correctness of Algorithm~\ref{alg:pdt to dt} described in the previous section, along with Observation~\ref{obs: markedanswersirrelevant}.
\end{proof}

\Cref{cor: indexing lifting} and \Cref{cor: ip lifting} follow from \Cref{thm: pdt depth lifting}, \Cref{thm: sdt depth lifting}, \Cref{thm: pdt size lifting}, the fact that $\IND_m$ is $m$-stifled (Claim~\ref{claim: ind property p}) and the fact that $\IP_m$ is $1$-stifled (\Cref{claim: ip stifled}).

\subsection{Non-adaptive and round-respecting simulations}

When run on a non-adaptive parity decision tree for $f \circ g$ of cost $d$, our simulation algorithm actually results in a non-adaptive decision tree for $f$ of cost $d/k$. This observation relies on the fact that in the course of every iteration (of Line~\ref{line: simulationloop} of Algorithm~\ref{alg:pdt to dt}), the updates done to $M, \MARK, Q$ and $\FREE$ depend solely on the sequence of parities $P$ that are processed. More formally,

\begin{itemize}
    \item The outputs $P'$ and $j$ from the $\rowreduce$ routine are functions of $P,M,\MARK,k$ and $\FREE$. They \emph{do not} depend on the input $A$. $M$ and $\MARK$ are then updated depending solely on these outputs.
    \item The decision to query a variable $x_j$ (and hence to append $j$ to $Q$, in Lines~\ref{line: query} and~\ref{line: updateq} of Algorithm~\ref{alg:pdt to dt}) depends only on $\MARK$, and the resulting update to $\FREE[j]$ (Line~\ref{line: update free} and the two preceding lines of Algorithm~\ref{alg:pdt to dt}) only depends on $M$ and $\MARK$.
\end{itemize}

Now consider a run of the simulation algorithm on a non-adaptive PDT. By Observation~\ref{obs: cost} this results in a DT that computes $f$ with cost at most $d/k$. By the observations above and since the parity decision tree is non-adaptive, the simulation algorithm processes the same sequence of parity queries regardless of what path it takes down the non-adaptive PDT. Hence the set $Q$ (and hence the queries made) does not depend on what path the simulation algorithm takes, and they can be made non-adaptively. Thus, we obtain the following theorem where $\NADT(\cdot)$ denotes non-adaptive decision tree complexity and $\NAPDT(\cdot)$ denotes non-adaptive parity decision tree complexity.
\begin{theorem}\label{thm: nonadaptive}
Let $g:\{0,1\}^m \to \{0,1\}$ be a $k$-stifled function, and let $f \subseteq \zone^n \times \R$ be a relation. Then,
\[
\NAPDT(f \circ g) \geq \NADT(f) \cdot k.
\]
\end{theorem}

This argument can easily be adapted to say that when run on a `$t$-round' PDT (that is, parity queries are done simultaneously in each round), the simulation algorithm results in a $t$-round DT with the guarantee that if the parity decision tree has made at most $d_i$ queries during its first $i$ rounds, the decision tree makes at most $d_i/k$ queries during its first $i$ rounds.

\bibliography{bibo}

\appendix

\section{Deferred proofs}\label{app: stifled proofs}

\begin{claim}\label{claim: sdt lifting implies depth-size lifting}
    \Cref{thm: pdt size lifting} follows (up to a constant factor in the exponent in the RHS) from \Cref{thm: sdt depth lifting}. \Cref{thm: sdt depth lifting} follows (up to a $\log n$ factor in the RHS) from \Cref{thm: pdt size lifting}.
\end{claim}

This claim immediately follows from the following bounds relating $\PDTsize$ and $\sDT$.

\begin{claim}\label{claim: pdtsize sdt relationship}
    Let $F \subseteq \zone^n \times \R$ be a relation. Then,
    \[
    \frac{\log \PDTsize(F)}{\log n} \leq \sDT(F) \leq 2\log\PDTsize(F).
    \]
\end{claim}

We show the first inequality by a naive simulation of a sDT by a PDT, and the second one via a standard tree/protocol-balancing trick.
\begin{proof}
We prove the inequalities in order.
\begin{itemize}
    \item Consider a sDT $T$ of depth $d$ computing $F$. Consider the following PDT computing $F$: Start from the root of $T$ and traverse down until a leaf in the following way. For a subspace query encountered at a node, query its parities in order until encountering a parity that violates the subspace constraint. At this point go to the 0-child of the current node in $T$ and recurse. If no such parity is encountered then go to the 1-child of the current node in $T$ and recurse. Note that the 0-child of a node can be reached in at most $n$ different ways (one for each constraint of the subspace, and since the co-dimension of a subspace is at most $n$) and the 1-child can be reached in exactly one way. Iterating until we reach a leaf, we obtain a PDT of size at most $n^d$. Thus,
    \[
    \log \PDTsize(F) \leq \sDT(F)\log n.
    \]
    \item Consider a PDT $T$ of size $s = \PDTsize(F)$ computing $F$. We first claim that there is an internal node $v$ in $T$ such that the number of leaves $\ell$ in the subtree rooted at it is at least $s/3$ and at most $2s/3$. To see this, first identify a node $v'$ in $T$ with the smallest number of leaves under it among all nodes with at least $2s/3$ leaves under it. The child of this node with the larger subtree has at least $s/3$ leaves and at most $2s/3$ leaves (by the minimality of $v'$). A subspace DT protocol for $F$ is as follows: Determine whether the input reaches $v$ or not. This can be done with one subspace query since it involves checking whether a set of affine parity constraints are all satisfied (and this is precisely the indicator of an affine subspace). If the input reaches $v$, then continue in the tree rooted at $v$. In the other case, we can eliminate the tree rooted at $v$ from $T$. In either case, the number of leaves in the resultant PDT is at most $2s/3$. Recursively using the same protocol until we reach a leaf, this gives a subspace DT computing $F$ of cost $\log_{3/2}s \leq 2\log s$.
\end{itemize}
\end{proof}

We next state a claim that shows that the randomized communication complexity of a function is bounded from above by its PDT size. Let $\mathsf{R}^{cc}_{\epsilon}(\cdot)$ denote $\epsilon$-error randomized communication complexity. It is well known that the public-coin $\epsilon$-error randomized communication complexity of the Equality function on $2n$ input bits is $O(\log(1/\epsilon))$. It is easy to see that the same protocol and same upper bound also holds for the communication complexity of checking whether a joint input lies in an arbitrary (but known to both players) affine subspace.
\begin{claim}\label{claim: randomized lifting implies dt-pdtsize lifting}
Let $F \subseteq \zone^n \times \zone^n \times \R$ be a relation. Then,
\[
\mathsf{R}^{cc}(F) \leq O(\log(\PDTsize(F)) \times \log \log (\PDTsize(F))).
\]
\end{claim}
\begin{proof}
Consider a PDT $T$ of size $s$ computing $F$. Just as in the previous proof, Alice and Bob jointly identify a node $v'$ in $T$ that has at least $s/3$ and at most $2s/3$ leaves under it. They use $O(\log \log s)$ bits of communication to find out with error probability at most $1/(6 \log s)$ whether the computation of $T$ on their joint input reaches $v'$. This can be done since the set of inputs reaching $v'$ is an affine subspace. Irrespective of the outcome, the size of the resultant PDT reduces to at most $2s/3$. The players then recurse at most $2 \log s$ times to reach a leaf and output the value at that leaf. By a union bound, the error of this protocol is at most $1/3$. The total cost of this protocol is at most $O(\log s \log \log s)$.
\end{proof}
In particular, the claim above implies that a lifting theorem from randomized decision tree complexity of $f$ to $\mathsf{R}^{cc}(f \circ g)$ (see, for example,~\cite{GPW20, CFKMP21}, for such lifting theorems, albeit for gadgets of super-constant size) implies a lifting theorem from randomized decision tree complexity of $f$ to $\PDTsize(f \circ g)$.

\begin{proof}[Proof of \Cref{claim: ind property p}]
Let $S$ be an arbitrary subset of $[m] \sqcup [2^m]$ of size $m$, and let $b \in \zone$. We need to show that there exists a setting of variables in $[m] \sqcup [2^m] \setminus S$ that fixes the value of $\IND_m$ to $b$. Suppose there are $\ell$ addressing variables and $m - \ell$ target variables in $S$. Fix the remaining $2^m - m + \ell$ target variables to value $b$. 

We now claim that there exists a setting of the free (outside of $S$) $m - \ell$ addressing variables such that the `relevant' target variable is always one of the free (outside of $S$) target variables, regardless of the setting of the $\ell$ `constrained' addressing variables in $S$. This would prove the claim since we have set all of the free target variables to value $b$, and hence $\IND_m$ would always output $b$ regardless of the assignments to variables in $S$. Towards a contradiction, suppose not. Then, for each setting of the free $m - \ell$ addressing variables there exists a setting of the constrained addressing variables such that the relevant target variable is one of the $m - \ell$ constrained target variables. Thus,
\[
2^{m - \ell} \leq m - \ell,
\]
which is easily seen to be false for all $\ell \in \cbra{0, 1, \dots, m-1}$ and all integers $m \geq 1$.
\end{proof}

\begin{proof}[Proof of \Cref{claim: ip stifled}]
Let the inputs to $\IP_m$ be denoted by $x_1, x_2, \dots, x_m, y_1, y_2, \dots, y_m$. Pick an arbitrary set of size 1. Without loss of generality we may assume that this set is $\cbra{x_1}$. Given $b \in \zone$ we need to show that there exists a setting of the variables $x_2, \dots, x_m, y_1, \dots, y_m$ that fixes the value of $\IP_m$ to $b$. To this end, fix $y_1 = 0$. This implies $x_1 \wedge y_1 = 0$. Hence, regardless of the value of $x_1$, the function now evaluates to $\IP_{m-1}(x_2, \dots, x_m, y_2, \dots, y_m)$. Since $m \geq 2$, $\IP_{m-1}$ is a non-constant function, and we can set the values of $x_2, \dots, x_m, y_2, \dots, y_m$ such that $\IP_{m-1}(x_2, \dots, x_m, y_2, \dots, y_m)$ (and thus $\IP_m(x_1, x_2, \dots, x_m, y_1, y_2, \dots, y_m)$) evaluates to $b$. This concludes the proof.
\end{proof}

\begin{remark}
The parameters in \Cref{claim: ind property p} and \Cref{claim: ip stifled} cannot be improved.
That is,
\begin{itemize}
    \item $\IND_m$ is not $(m+1)$-stifled: Consider the set $S$ of all addressing variables and a single target variable. We have $|S| = m + 1$. Regardless of how we set the remaining variables, the restricted function is never fixed because the addressing variables could be assigned values to point to the target variable in $S$ (which could have value either 0 or 1).
    \item $\IP_m$ is not $2$-stifled: Consider the set $x_1, y_1$. Regardless of how we set the remaining variables, the function value is different when $x_1 \wedge y_1 = 1$ as compared to when $x_1 \wedge y_1 = 0$.
\end{itemize}
\end{remark}

We observe below that `Majority is stiflest'. More precisely, we show that $\MAJ_n : \zone^n \to \zone$ is $(\lceil n/2 \rceil - 1)$-stifled and no Boolean function on $n$ inputs is $\lceil n/2 \rceil$-stifled.
\begin{claim}[Majority is stiflest]
    Let $f : \zone^n \to \zone$ be a Boolean function. Then, $f$ cannot be $\lceil n/2 \rceil$-stifled. Moreover $\MAJ_n$ is $(\lceil n/2 \rceil - 1)$-stifled.
\end{claim}

\begin{proof}
We prove the two claims in order.
\begin{itemize}
    \item Towards a contradiction, assume $f$ is $\lceil n/2 \rceil$-stifled. Pick an arbitrary set $S$ of size $\lceil n/2 \rceil$. By definition, there exists a setting of variables in $\bar{S}$ that forces the value of $f$ to 0. Since $|\bar{S}| = \lfloor n/2 \rfloor \leq \lceil n/2 \rceil$, the assumption also implies the existence of a setting of variables in $S$ that forces the value of $f$ to 1. Since these settings are on disjoint sets of variables, this implies a setting of variables that forces the value of $f$ to both 0 and 1, which yields a contradiction.
    \item Pick an arbitrary set of variables $S$ of size $\lceil n/2 \rceil - 1$ and an arbitrary $b \in \zone$. By definition of the Majority function, fixing all variables in $\bar{S}$ to value $b$ forces the function value to be $b$, concluding the proof.
\end{itemize}
\end{proof}

\begin{proof}[Proof of \Cref{claim: random stifled}]
Fix $S \subseteq [m]$ with $|S| = k = \log m/2$. For $z \in \zone^{[m] \setminus S}$ and $z' \in \zone^S$, let the input $(z, z') \in \zone^m$ be defined by $(z, z')|_S = z'$ and $(z, z')|_{[m] \setminus S} = z$. We have for every fixed $z \in \zone^{[m] \setminus S}$,
\[
\Pr[\forall z' \in \zone^S, g(z, z') = 0] = \frac{1}{2^{2^k}}.
\]
Since these events are independent for each $z \in \zone^{[m]\setminus S}$,
\[
\Pr[\nexists z \in \zone^{[m] \setminus S} \text{ such that } \forall z' \in \zone^S, g(z, z') = 0] = \rbra{1 - \frac{1}{2^{2^k}}}^{2^{m - k}}.
\]
Thus, the probability that `$S$ isn't 0-stifled' is at most $\rbra{1 - \frac{1}{2^{2^k}}}^{2^{m - k}}$. The same bound holds for the probability that $S$ isn't 1-stifled. By a union bound over all $S \subseteq [m]$ with $|S| = k$, we get
\begin{align*}
\Pr[g \text{ is not $k$-stifled}] & \leq \binom{m}{k} \cdot 2\rbra{1 - \frac{1}{2^{2^k}}}^{2^{m - k}}\\
& \leq 2^{k \log m + 1} \cdot \exp\rbra{-2^{m - k - 2^k}}\tag*{using standard inequalities}\\
& \leq 2^{(\log^2m)/2 + 1 - 2^{(m - \log m - \sqrt{m}) \log e}}\\
& \leq 1/10,
\end{align*}
where the last inequality holds for sufficiently large $m$.
\end{proof}

\end{document}